\documentclass[conference]{IEEEtran}
\usepackage{cite}
\usepackage{amsmath,amssymb,amsfonts,amsthm}
\usepackage{graphicx}
\usepackage{textcomp}
\usepackage{xcolor}
\usepackage{xspace}

\usepackage{enumitem}
\usepackage{ipdps26repro}

\usepackage{algorithm}
\usepackage{algpseudocodex}

\newtheorem{lemma}{Lemma}
\newtheorem{theorem}{Theorem}

\usepackage{booktabs}
\usepackage{siunitx}

\title{GPU-Accelerated Multilevel Graph Clustering: A Parallel Perspective on Louvain and Leiden}
\author{\IEEEauthorblockN{Michael S. Gilbert}
\IEEEauthorblockA{\textit{Dept. of Computer Science and Engineering} \\
\textit{The Pennsylvania State University}\\
University Park, PA, USA \\
msg5334@psu.edu}
\and
\IEEEauthorblockN{Kamesh Madduri}
\IEEEauthorblockA{\textit{Dept. of Computer Science and Engineering} \\
\textit{The Pennsylvania State University}\\
University Park, PA, USA \\
madduri@psu.edu}
}

\begin{document}

\maketitle

\begin{abstract}
The sequential Louvain and Leiden algorithms are widely used techniques for modularity-optimizing clustering (or community detection) in large graphs.
We present pLouvain and pLeiden, two new GPU parallelizations. pLouvain is based on the Louvain+ extension. pLeiden is the first parallel implementation to provably preserve all quality guarantees of sequential Leiden. We achieve this through a novel spanning-tree-based refinement approach. Both pLouvain and pLeiden use a lightweight symmetry-breaking technique that emulates an ordered traversal of vertices. For pLouvain, we develop an alternative iteration strategy to rectify the weak internal cluster connectivity observed in Louvain/Louvain+. Further, both pLouvain and pLeiden optimize the LambdaCC objective function, a generalization of modularity and the related Constant Potts model.

On a collection of 57 graphs from 10 families, our results show that pLouvain and pLeiden achieve geometric mean speedups of 3.1x and 8.8x, respectively, over the current fastest open-source parallelizations of Louvain and Leiden. For the clusterings generated, pLouvain yields the highest modularity scores on nearly all tested graphs. The subroutines within these two multilevel approaches could aid in the parallelization of other Louvain-based techniques.

\end{abstract}

\section{Introduction}
Graph clustering is the task of identifying dense sub-structures in graphs.
The applications range from protein function analysis to recommendation systems \cite{fortunato2010community}.
Heuristic methods such as the Louvain \cite{louvain_base} and Leiden \cite{leiden_base} methods exploit the common hierarchical structure of real-world graphs.
Louvain and Leiden are both multilevel methods, a class of heuristics used to solve a variety of problems involving sparse systems~\cite{Teng99}.
Both methods have become foundational techniques in graph analysis.

Due to their popularity and the need to process large graphs, a plethora of prior work has probed their potential for parallelism.
Despite Louvain's almost two-decade long existence, one such recent work \cite{gve_louvain} greatly improved upon the efficiency of other parallel Louvain methods, by a factor of at least twenty times on multicore systems and even six times versus cuGraph's GPU Louvain.
A corresponding GPU implementation \cite{Sahu2025-ye} failed to generate significant performance improvements on many graphs over the multicore version.
Another recent GPU implementation \cite{wang_gpu} is slower still according to our experiments.
pLouvain, our GPU Louvain implementation, is the first to achieve substantial speedups versus the state-of-the-art multicore CPU implementations.

pLouvain is based on Louvain+ \cite{louvain+}, a simple extension to Louvain that further develops the multilevel aspect of Louvain to deliver superior quality.
A prior parallel adaptation~\cite{shi2021scalable} of Louvain+ targeted multicore CPUs using an asynchronous approach.
We find Louvain+ is particularly beneficial for a synchronous approach, which is necessitated by the high degree of concurrency on modern GPUs. pLouvain further includes several GPU-centric enhancements to the \textit{local move} and graph contraction phases of Louvain.

The Leiden method has seen fewer attempts at parallelization.
Among them, none provably offer the core guarantees that define the Leiden method.
We present a novel approach to Leiden's refinement scheme built upon spanning trees, which provides all the core guarantees.
Our proof of this claim observes that the original Leiden refinement scheme can generate a clustering if-and-only-if our parallelization pLeidenR can generate it.
Furthermore, we extend Leiden in the vein of Louvain+ to create pLeiden+.

In summary, the contributions of this work are as follows:
\begin{enumerate}
    \item On an NVIDIA B200 GPU and a 57-instance graph dataset, pLouvain achieves a 3.1x geometric mean speedup over the fastest competitor~\cite{Sahu2025-ye}.
    \item pLeiden provably ensures the quality guarantees of Leiden, and is 8.8x faster (geometric mean) than a prior parallelization \cite{leiden_sahu}.
    \item pLouvain and pLeiden+ consistently outperform other approaches in terms of quality (modularity).
    \item We present an efficient symmetry-breaking technique based on the Jet \cite{jet} graph partitioner's \textit{afterburner} filter.
    \item Our graph contraction implementation is $\geq$15x faster by geometric mean than state-of-the-art competitors \cite{Sahu2025-ye, wang_gpu}.
    \item We give an improved iteration scheme to mitigate Louvain's weak internal cluster connectivity problem \cite{leiden_base}.
\end{enumerate}

\section{Problem Definition}
Let $G = (V,E)$ represent an undirected graph, $V$ denoting the vertex set and $E$ denoting the edge set.
An edge is defined as a pair $(u,v)$, with vertices $u,v \in V$.
Both vertices and edges may have associated weights, denoted $w(v)$ and $w(u,v)$, respectively.
While a general definition of $w$ maps each object to $\mathbb{R}$, we consider only positive integer weighting functions for simplicity.
We define $w(u,v) = 0 \iff (u,v) \notin E$.
The \textit{neighborhood} of a vertex $v$ is the set of adjacencies of $V$, denoted as $N(v)$.
Let $C = (C_1,C_2,C_3,...,C_x)$ denote a clustering, with each cluster $C_i \subset V$.
Each cluster must be pairwise disjoint, and their union equal to $V$.
Let $C[v]$ denote the cluster to which vertex $v$ belongs.
We define $w(X) = \sum_{v \in X} w(v)$ for any vertex set $X$, and $w(X,Y) = \sum_{(u,v) \in X \times Y \setminus X} w(u,v)$ for any two sets of vertices $X,Y$.
Symmetry of undirected graphs implies $w(X,Y) = w(Y,X)$ when $X \cap Y = \emptyset$.
A singleton cluster is a cluster containing only one vertex, while an all-singleton clustering is a clustering consisting of only singleton clusters.
We define the \textit{cutset} of a clustering $C$ as the subset of $E$ such that the two vertices of an edge are in different clusters of $C$.
We define the \textit{envelope} of $C$ (notated $env(C)$) as the subset of $E$ such that both vertices are in the same cluster, with $\textit{cutset}(C) \cup env(C) = E$.

\subsection{LambdaCC}
Graph clustering is typically formulated as an optimization task for one of many possible objective functions.
Many of the most significant objective functions are unified by the lambdaCC objective \cite{lambdacc}:
\begin{equation}
\label{eq:lambdacc}
\begin{split}
    \lambda CC(C) &= \sum_{C_i \in C} \sum_{\substack{u,v \in C_i \\ u \neq v}} (w(u,v) - \lambda w(u)w(v)) \\
                 &= \sum_{(u,v) \in env(C)}w(u,v) -\lambda \sum_{C_i \in C} \sum_{\substack{u,v \in C_i \\ u \neq v}} w(u)w(v)
\end{split}
\end{equation}
It rewards edges within a cluster, while penalizing every pair of vertices within each cluster.
$\lambda$ balances the influence of the reward versus the penalty.
It is in the range $[0,1]$, and may be a function of graph parameters such as $|E|$ or $|V|$.
While the form presented in the original lambdaCC paper \cite{lambdacc} is a minimization objective, we use an equivalent maximization objective.
Maximizing $\lambda CC(C)$ on $G$ is equivalent to maximizing $\sum_{C_i \in C} \sum_{{\substack{u,v \in C_i \\ u \neq v}}} w'(u,v)$ on a \emph{lambdaCC signed graph}. This is a fully connected, undirected graph with edge weights given by $w'(u,v) = w(u,v) - \lambda w(u)w(v)$.  
$w'(u,v) > 0$ only if $(u,v) \in E$. 

$w'$ is a useful concept for defining clustering properties and objective deltas.
We will often use $w'$ with respect to vertex sets $X,Y$: $w'(X,Y) = w(X,Y) - \lambda w(X)w(Y \setminus X)$.
$2w'(v,C[v])$ gives the individual contribution of a vertex $v$ to the overall objective,
while $2w'(v,C_i) - 2w'(v,C[v])$ gives the objective delta related to moving vertex $v$ to cluster $C_i$.

With proper choice of $w$ and $\lambda$, one can obtain objectives including modularity \cite{modularity}, the Constant-Potts Model (CPM) \cite{cpm_traag}, the Absolute-Potts Model (APM) \cite{BRSLLP}, label propagation ($\lambda = 0$) \cite{raghavan2007near}, sparsest-cut, cluster deletion, and connected components ($\lambda = \epsilon$). Specifically, for modularity, $w(v) = |N(v)|$, $(u,v) \in E \iff w(u,v) = 1$, and $\lambda = \frac{\gamma}{2|E|}$; $\gamma$ is the modularity resolution parameter.
It is NP-Hard to optimize or approximate this objective for general $w$ and $\lambda$ \cite{lambdacc}, which it inherits from modularity \cite{modularity_npcomplete} (where these are only NP-Complete) via a reduction.

\subsection{Leiden Guarantees}
Leiden provides six guarantees at different timescales, and defines each relative to an objective function.
References to $\gamma$ in each property's name are related to the resolution parameter of the Constant Potts Model, but each term can be equivalently defined relative to the lambdaCC signed graph.\\
\textit{$\gamma$-separation}: For any two clusters $C_x,C_y$ in a clustering, $w'(C_x,C_y) \leq 0$.\\
\textit{$\gamma$-connectivity}: Every cluster $C_x$ can be decomposed into a binary tree of vertex-sets, such that every set is a subset of its parent and the union of its children. Every leaf node represents a unique vertex in $C_x$. For every node with two children $X_1,X_2$, $w'(X_1,X_2) \geq 0$.\\
\textit{Subpartition $\gamma$-density}: $\gamma$-connectivity with one extra condition: every node $X$ in the tree must satisfy $w'(X,C_x) \geq 0$.\\
\textit{Node optimality}: For every vertex $v$ and cluster $C_x$, $w'(v, C[v]) \geq w'(v,C_x)$ and $w'(v, C[v]) \geq 0$.\\
\textit{Uniform $\gamma$-density}: For every cluster $C_x$ and every subset $X$ of that cluster, $w'(X,C_x) \geq 0$.\\
\textit{Subset Optimality}: For every pair of clusters $C_x,C_y$ and every subset $X$ of $C_x$, $w'(X,C_x) \geq w'(X,C_y)$ and $w'(X,C_x) \geq 0$.

This concise restating of Leiden guarantees in lambdaCC notation is a new contribution of this paper.

\section{Background}
\subsection{Multilevel Clustering Algorithms}
Heuristic algorithms are the predominant approach for optimizing lambdaCC-related objectives on real-world graph data, due to its computational hardness.
The Louvain algorithm \cite{louvain_base} is one of the most widely used algorithms, as it produces high-quality clusterings in linear time.
Louvain exploits the common hierarchical structure of clusters in real-world networks via a a multilevel algorithm.
First, a \emph{local move heuristic} greedily optimizes an initial all-singleton clustering.
This local move heuristic is computationally equivalent to running the label propagation algorithm \cite{raghavan2007near} on the lambdaCC signed graph, up to the choice of initial clustering state and vertex traversal order.
Second, it generates a coarsened graph from the optimized clustering, such that each coarse vertex represents an entire cluster.
Each coarse vertex $x$ corresponding to cluster $C_x$ has weight equal to $w(C_x)$.
Let $x,y$ represent two coarse vertices corresponding to clusters $C_x,C_y$, respectively.
There are edges $(x,y)$ and $(y,x)$ in the coarse graph with weight $w(C_x,C_y)$, if and only if $w(C_x,C_y) > 0$.
This process terminates when the clustering is $\gamma$-separated.
Every clustering on a coarsened graph induces a clustering on the input graph;
the objective on each coarsened graph is monotonic with respect to the objective on the input graph, and the difference is a constant.
The original Louvain algorithm avoids a delta between coarse and input objectives by using weighted self-loops for every coarse vertex.
This is unnecessary; instead, one may sum the constant objective deltas between each successive coarse graph.

One notable extension of Louvain, Louvain+ \cite{louvain+}, adds an \textit{uncoarsening} phase, which applies the local move heuristic again to the result of each recursive call to further improve the clustering.
Refer to Algorithm~\ref{alg:louvain} for an overview.
VieClus \cite{BiedermannHSS18}, a memetic algorithm for clustering, and a multicore CPU parallel work \cite{shi2021scalable} implement this innovation.
When discussing our implementation of Louvain+, we refer to the \textit{coarsening} phase, which is everything that occurs prior to the uncoarsening phase.
Separately, an analysis \cite{multilevel_cluster} of multilevel clustering variants verified the importance of local moving during the uncoarsening. 

The Leiden algorithm \cite{leiden_base} is another improvement upon Louvain, motivated by the observation that Louvain generates poorly internally-connected clusters.
To address this, it applies a refinement heuristic after the local move heuristic.
This refinement is similar to \emph{solution-based coarsening} schemes from graph partitioning \cite{Walshaw2004-ek}: it exclusively generates sub-clusters of those given by the local move heuristic.
The output of the refinement is used for the graph coarsening operation, and the output of the local move heuristic is used to initialize the clustering of the coarse graph, instead of a singleton clustering.
Leiden can be fed any input clustering as a starting point; by feeding itself its own output, it can be \emph{iterated}.
This is also true of Louvain, but the Leiden authors claim it rarely benefits from more than one additional iteration \cite{leiden_base}.
We will later demonstrate an alternative way to iterate Louvain that is superior to the one considered by the Leiden authors.

After every iteration, Leiden guarantees $\gamma$-separability and $\gamma$-connectivity of the clustering.
After a stable iteration (one which does not improve the clustering), Leiden guarantees subpartition $\gamma$-density and node optimality.
It guarantees uniform $\gamma$-density and subset optimality asymptotically, via randomization.

\begin{algorithm}[htbp]
\caption{Overview of the Louvain and Louvain+ algorithms.}
\label{alg:louvain}
\begin{algorithmic}[1]
\Require Undirected $G(V,E) = G_0$. $n$ = $|V|$.
\Ensure Vector $C_{out}$ mapping $V$ to $\{1, 2, ..., c\}$ for some integer $c \leq |V|$.
\State $l \gets 0$
\While {\textsc{True}}
    \State $C_l \gets $ \textsc{SingletonClustering($V_l$)}
   \State  $C_l \gets$ \textsc{LocalMovePasses}($G_l$, $C_l$)
   \If {$|C_l| = |V_l|$}
   \State {Break out of while loop} 
   \EndIf
    \State $G_{l+1} \gets$ \textsc{GraphContraction}($G_l$, $C_l$)  
    \State $l \gets l + 1$
\EndWhile
\While {$l - 1 \geq 0$} 
\State $l \gets l-1$
    \State $C_l \gets$ \textsc{ProjectCluster}($C_l, C_{l+1}$)
    \State  $C_l \gets$ \textsc{LocalMovePasses}($G_l$, $C_l$) \Comment{only Louvain+}
\EndWhile
\State $C_{out} \gets C_0$

\end{algorithmic}
\end{algorithm}

\subsection{Parallel Implementations}
Most research on parallel Louvain methods focuses on the local move heuristic.
The local move heuristic presents two primary challenges \cite{d2_is_louvain} for parallel implementations that motivate a variety of approaches.
These challenges are: 1) race conditions due to vertex state updates, and 2) destructive interactions within a set of simultaneous vertex moves that reduce its cumulative objective delta.
Synchronous approaches avoid race conditions, but require symmetry-breaking techniques to mitigate destructive interactions.
Bulk synchronous parallel (BSP) \cite{distributed_louvain_1} approaches process all vertices simultaneously and do not update vertex states until after processing every vertex.
Such approaches use explicit symmetry-breaking heuristics, like the minimum label heuristic (MLH) \cite{gve_louvain, Sahu2025-ye, forster_gpu}.
Some synchronous approaches \cite{batched_louvain, wang_gpu, grappolo, d2_is_louvain} use batches; vertices within a batch are concurrently processed, one batch at a time\footnote{In \cite{batched_louvain}, multiple batches are concurrently processed, one batch per GPU }, and apply state updates between batches.
Techniques for batching include distance-1 coloring \cite{grappolo}, distance-2 coloring \cite{d2_is_louvain}, degree-based bucketing \cite{wang_gpu}, and partitioning \cite{batched_louvain}.
Batches provide an implicit form of symmetry-breaking, but limit the potential degree of parallelism.
The different batching tactics and explicit symmetry-breaking heuristics address the problem of destructive interactions to varying degrees.

A third approach is asynchronous parallel \cite{Sahu2025-ye, gve_louvain, shi2021scalable}, wherein vertex states are updated as they are processed.
Prior work has shown that asynchronous algorithms converge more quickly than synchronous algorithms \cite{shi2021scalable}, presuming a CPU-level degree of concurrency.
Unfortunately, asynchronous algorithms are subject to race conditions, which lead to quality losses and slower convergence with increasing concurrency.
They are therefore ill-suited to GPUs, and our evaluation of such approaches in the empirical results section shows this.

Several works \cite{Sahu2025-ye, batched_louvain, forster_gpu, wang_gpu, grappolo_gpu} have examined parallel Louvain for GPUs.
Each of these approaches utilizes one of the aforementioned schemes.
We detail the two most efficient approaches of these, v-Louvain \cite{Sahu2025-ye} and GALA \cite{wang_gpu}.
v-Louvain utilizes a asynchronous approach for the local move heuristic, while GALA uses a batch-synchronous approach.
Both group vertices into buckets by their degree, to enable processing with specialized kernels.
v-Louvain uses two buckets, and processes low-degree vertices with a single thread, and high-degree vertices with entire thread blocks.
GALA uses more buckets with greater granularity; its kernels are more specialized towards each bucket in terms of thread-block size, warp-level primitive usage, and shared vs global memory usage.
For symmetry breaking, v-Louvain integrates the MLH in every fourth pass, whereas GALA relies only on batching.
GALA additionally implements a pruning scheme to reduce the total amount of work in each pass.
Regarding graph contraction, both first gather vertices by cluster id (GALA does not mention this in their paper, but we confirm via inspection of their code).
We discuss the detriments of this preprocessing in section \ref{section:our_contraction}.

The Leiden algorithm is significantly more challenging to implement in parallel, due to its quality guarantees.
We highlight that no existing parallel implementation provides each of the original quality guarantees.
GVE-Leiden's \cite{leiden_sahu} parallel version of Leiden's refinement has race conditions which lead to violations of the guarantees.
A prior implementation (NetworKit Leiden) \cite{nguyen_leiden} locks entire clusters to avoid such a scenario, which severely restricts the attainable degree of parallelism;
the author of that work stated that speedups beyond 32 threads on a 128-core system were negligible.
Furthermore, its refinement implementation lacks randomization of the cluster joining operation, so it fails to guarantee subset optimality and uniform $\gamma$-density asymptotically.
While GVE-Leiden's authors observed that NetworKit Leiden produced disconnected clusters \cite{leiden_sahu}, this was fixed in v11.2.

\section{Parallel Louvain}
We utilize the modified Louvain structure explored in \cite{louvain+}, which we outline in Algorithm \ref{alg:louvain}.
Line 12 represents the primary distinction between Louvain+ and Louvain.
There are two reasons we choose this approach:
\begin{enumerate}
    \item The quality as measured by modularity is improved versus standard Louvain \cite{louvain+, shi2021scalable}.
    \item The local move heuristic is cheaper during the uncoarsening phase. There are far fewer clusters to consider, and a smaller proportion of vertices are not node optimal.
\end{enumerate}
There is an opportunity to replace some of the expensive passes of the local move heuristic during the coarsening with cheaper passes during the uncoarsening.
This is especially important for synchronous approaches such as the one we outline, which require more passes to converge than asynchronous approaches favored by CPU-oriented algorithms.
Each of our algorithms assumes the graph is stored in memory in the compressed sparse row (CSR) format.

\subsection{Local Move Heuristic}
We implement the local moving phase of the Louvain algorithm with a bulk-synchronous design.
See Algorithm~\ref{alg:local_move} for a sketch of our approach.

\subsubsection{Afterburner Filter}
We adapt the Jet label propagation algorithm \cite{jet} to consider lambdaCC objectives.
It integrates a symmetry-breaking technique called the \emph{afterburner filter} (see algorithm \ref{alg:after_filter}), which deprioritizes destructive interactions between adjacent vertices and prioritizes constructive interactions.
In comparison, the simpler MLH entirely forbids simultaneous moves which destructively interact via their adjacencies.
By considering constructive interactions, the afterburner filter can move vertices which are in a locally-optimal cluster (governed by the $\phi$ parameter), for potential global objective improvements.
Most significantly, as our empirical evaluation shows, this enables some basic simulated annealing techniques as in \cite{dkaminpar_jet} and \cite{detJetHyper}.
The afterburner filter does not consider interactions between non-adjacent vertices, a common limitation of the symmetry-breaking heuristics of which we are aware.

The ordering referenced on line 5 is a function of the precomputed objective deltas for each candidate move: greater objective deltas are earlier in the ordering.
If two candidate moves have objective deltas that differ by less than some $\epsilon$ (which we set as 0.1), then tiebreaking is performed by a hash function of the vertex ids.
The loop beginning on line 4 specifies a thread-block level reduction on the $\delta_v$ variable.

\begin{algorithm}[htbp]
\caption{pLouvain: Parallel local move heuristic.}
\label{alg:local_move}
\begin{algorithmic}[1]
\Require $G = (V,E)$. Vector $C_{in}$. Temperature parameter $\phi$.
\Ensure Vector $C_{out}$.
\State $C \gets C_{in}$, $C_{out} \gets C_{in}$
\If{$C$ is an all-singleton clustering}
\State $DS \gets G$ \Comment{$DS$ tracks $w(v,u)$}
\Else
\State $DS \gets $ \textsc{Enumerate}($G,C$) \Comment{$DS$ tracks $w(v,C_x)$}
\EndIf
\For{$i = 1,\dots,\mathit{limit}$}
\State $L \gets$ [] \Comment{empty movelist}
\State $D \gets $ [null]*$|V|$ \Comment{destination clusters}
\For{$v \in V$ \textbf{in parallel}}
\State $D[v] \gets $ argmax $C \setminus C[v]$ of $w'(v, C_i)$
\If{$w'(v,C[v]) < 0$ \textbf{and} $w'(v,D[v]) < 0$}
\State $D[v] \gets \emptyset$ \Comment{new singleton cluster}
\EndIf
\If{$w'(v,D[v]) \geq (1 - \phi_i) w'(v,C[v])$}
\State append $v$ to $L$
\EndIf
\EndFor
\State $L \gets $ \textsc{AfterburnerFilter}($G,L,C,D$)
\If{$|L| > 0.05 |V|$ and $C$ not an all-singleton clustering}
\State $DS \gets $ \textsc{UpdateAffected}($G,L,C,D$)
\Else
\State $DS \gets $ \textsc{RecomputeAffected}($G,L,C,D$)
\EndIf
\State $C \gets $ \textsc{Update}($L,C,D$) \Comment{Apply $D$ for $L$}
\If{$\lambda CC(C) > \lambda CC(C_{out})$}
\State $C_{out} \gets C$
\EndIf
\EndFor
\end{algorithmic}
\end{algorithm}

\begin{algorithm}[htbp]
\caption{The ``afterburner'' filter used in local move.}
\label{alg:after_filter}
\begin{algorithmic}[1]
\Require $G = (V,E)$. Candidate moves $L$. Current clustering $C$. Destination clusters $D$.
\Ensure Filtered moves $M$.
\State $M \gets$ []\Comment{empty movelist}
\For{$v \in L$ \textbf{in parallel}}
\State $\delta_v \gets w'(v,D[v]) - w'(v,C[v])$ \Comment{Precomputed}
\For{$u \in N(v)$ \textbf{in parallel}}
\If{$u \in L$ and $ord(u) < ord(v)$}
\If{$C[v] = C[u]$}
\State $\delta_v \gets \delta_v + w'(u,v)$
\ElsIf{$D[v] = C[u]$}
\State $\delta_v \gets \delta_v - w'(u,v)$
\EndIf
\If{$C[v] = D[u]$}
\State $\delta_v \gets \delta_v - w'(u,v)$
\ElsIf{$D[v] = D[u]$}
\State $\delta_v \gets \delta_v + w'(u,v)$
\EndIf
\EndIf
\EndFor
\If{$\delta_v \geq 0$}
\State append $v$ to $M$
\EndIf
\EndFor
\end{algorithmic}
\end{algorithm}

\subsubsection{Kernel Fission}
The most costly subtask of the local move heuristic is the enumeration of adjacent clusters and respective $w(v,C_i)$ for each vertex, followed closely by the argmax operation to select the best destination cluster.
The standard approach across several recent works \cite{Sahu2025-ye, wang_gpu} fuses these into one kernel.
We differ from these other works by employing kernel fission \cite{hijma2023optimization} to create separate kernels for the enumeration and argmax operations.
This enables further tuning and optimizations on a per kernel basis, which are incompatible with a single fused kernel.
There is a small tradeoff of one additional sequential read from global memory in the argmax kernel.

\subsubsection{Enumeration Kernel}
As in \cite{Sahu2025-ye}, we maintain a hash table in global memory for each vertex, with size given by its degree.
The total capacity among all hash tables is equal to $|E|$;
similar allocations are common for graph partitioning \cite{jet, terapart}.
We designate a special memory location for each vertex to accumulate its connection strength to its current cluster.
This enables the use of a reduction operation to accumulate this value, in place of high-contention atomic operations.
For vertices within a viable range of degrees, we build the hashtables in shared-memory before copying them to global memory;
all other vertices use global memory exclusively.

After each pass, we recompute the hash tables of each moved vertex and its adjacencies.
If the number of vertex moves is small (less than 3-5\% of all vertices), it is beneficial to instead modify the hash tables of the adjacent vertices directly in global memory.
This specialized kernel can't be fused with the argmax kernel.
It is most relevant to the uncoarsening phase, which usually sees only a small fraction of all vertices undergo a move.

For the first pass of the local move heuristic during the coarsening, the adjacent clusters for each vertex will be exactly equal to its adjacent vertices.
Thus, in the first pass only, we skip the enumeration kernel and pass the input graph directly to the argmax kernel.

\subsection{Graph Contraction}
\label{section:our_contraction}
We adopt an approach to parallel graph contraction defined in a recent graph partitioning work \cite{jet}.
This approach is given in Algorithm~\ref{alg:contraction}.
The central idea is to create hash tables for each coarse vertex, and then for each edge in the input graph, insert the corresponding coarse edge into the appropriate hash table.
Edges inducing self-loops in the coarse graph are ignored, as cluster sizes are tracked separately.
The sum of degrees of the constituent input vertices for each coarse vertex gives a loose upper bound for its hash table size.
For typical input graphs, the number of insertions into each hash table is at least one order of magnitude smaller than this upper bound.
This upper bound tightens with each coarsening phase.

To extract the edges of $G_c$ from the hash tables, we perform a single stream-compaction operation on the unified hash table arrays to select non-null entries.
Before this, we need to count the number of such entries to allocate memory for $G_c$, and on a per-coarse-vertex basis to compute the coarse offsets.
We integrate the counting of unique entries into the hashtable insertion operation.

In contrast to v-Louvain \cite{Sahu2025-ye} and GALA \cite{wang_gpu}, we don't gather fine vertices into groups by coarse vertex membership.
Gathering vertices in this way is common for sequential contraction schemes, where it enables the stream-compaction operation to occur concurrently with the hashtable operations.
In a non-sequential context, this optimization isn't possible.
The prior GPU works allocate work units per coarse vertex;
a work unit can be a single thread, a warp, a thread-block, or a full thread-grid, depending on the sum of input degrees.
Compared to our approach, in which we allocate work units per fine vertex, this exacerbates the load imbalance between work units, and reduces the potential for latency hiding.

\begin{algorithm}[htbp]
\caption{The parallel graph contraction algorithm.}
\label{alg:contraction}
\begin{algorithmic}[1]
\Require $G = (V, E)$. The cluster array $C$. Cluster count $n_c$.
\Ensure $G_c = (V_c,E_c)$
\State bound $\gets [0]*(n_c + 1)$
\For{$v \in V$ \textbf{in parallel}}
    \State bound[$C[v]$] $\gets$ bound[$C[v]$] + $|N(v)|$
\EndFor
\State offset $\gets $ \textsc{ExclusivePrefixSum}(bound)
\State $H_{k} \gets [\textnormal{null}]*|E|$ \Comment{per-vertex hash tables}
\State $H_{v} \gets [0]*|E|$
\For{$v \in V$ \textbf{in parallel}}
    \State $h_{k} \gets $ $H_{key}$[offset[$C[v]$]..offset[$C[v]+1$]]
    \State $h_{v} \gets $ $H_{val}$[offset[$C[v]$]..offset[$C[v]+1$]]
    \For{$u \in N(v)$ \textbf{in parallel}}
        \If{$C[v] \neq C[u]$}
        \State $i \gets $ \textsc{InsertOrLookup}($h_{k}$, $C[u]$)
        \State $h_{v}[i]$ $\gets$ $h_{v}[i]$ + $w(u,v)$
        \EndIf
    \EndFor
\EndFor
\State $G_c \gets $ \textsc{StreamCompaction}($H_{k}, H_{v}$)
\end{algorithmic}
\end{algorithm}

\subsection{Miscellaneous Detail}

\subsubsection{Memory Reuse}
In order to reduce the number of large memory allocations and deallocations, we reuse memory where possible.
The memory for the hash tables needed for the local move heuristic is preserved and reused for all subsequent calls, as the coarser graphs will never need more memory for this than the input graph.
Additionally, we reuse this same memory for the hash tables in the graph contraction operation.
The only memory which must be allocated dynamically is for storing each coarse graph; all other auxiliary memory can be allocated in advance.

\subsubsection{Degree-Specialized Kernels}
Like v-Louvain \cite{Sahu2025-ye} and GALA \cite{wang_gpu}, we group vertices into buckets by degree for processing by kernels tuned to each bucket.
We use three buckets, for low, medium, and high degree vertices.
We specialize the enumeration, argmax, afterburner, and contraction kernels, with the argmax and afterburner kernels using higher thresholds for each bucket than the other two.

\subsection{Note on Iterating Louvain+}

We detail a new method for iterating Louvain, using Louvain+.
The old method, detailed by the Leiden authors \cite{leiden_base}, iterates Louvain by feeding it an input clustering, which in turn is fed into the local move heuristic in place of an all-singleton clustering.
VieClus \cite{BiedermannHSS18} implies an alternate approach, but we do not believe this was explicitly detailed.
In that work, their ``Multi-level Recombination Operator'' combines two clusterings $B_1$ and $B_2$, using a constrained version of Louvain+.
Its local move heuristic begins from an all-singleton clustering, but it may only create clusters $C_x$ that satisfy $B_1[u]=B_1[v]$ and $B_2[u]=B_2[v]$ for every vertex pair $u,v \in C_x$.
An uncoarsening phase allows for inter-constraint-cluster vertex movements.

Our new iteration method is similar, with one constraint clustering in place of two.
It differs in that the constraint clustering is not further used to give an initial clustering of the coarsest graph.
The new method is conceptually similar to iterated multilevel graph partitioning methods \cite{Walshaw2004-ek}.
While the Leiden paper showed that iterating Louvain in the former way worsens the internal connectivity of clusters,
our new approach guarantees connected clusters in stable iterations.
In a sequential setting, the coarsening phase alone always generates a clustering at least as good as the constraint.
This is due to the coarsening termination condition: any two clusters in the coarsening output satisfy $w'(C_x,C_y) \leq 0$, unless $C_x$ and $C_y$ are in different constraint clusters.

\section{Parallel Method for Leiden Guarantees}

Algorithm~\ref{alg:leiden} gives an overview of pLeiden, our novel Leiden parallelization, and can be compared to Algorithm~\ref{alg:louvain}.

\begin{algorithm}[htbp]
\caption{pLeiden: Overview of our parallel algorithm.}
\label{alg:leiden}
\begin{algorithmic}[1]
\Require $G = (V,E)$. Vector $C_{in}$.
\Ensure Vector $C_{out}$.
\State $l \gets 0, C_0 \gets C_{in}$
\While {\textsc{True}}
    \State $\delta \gets \lambda CC(C_l)$ 
   \State  $C_l \gets$ \textsc{LocalMovePasses}($G_l$, $C_l$)
   \If{$\lambda CC(C_{l}) = \delta$}
\State $C_l \gets $ \textsc{LocalMoveAlt}($G,C_l$)
\EndIf
   \If {$|C_l| = |V_l|$}
   \State {Break out of while loop} 
   \EndIf
\State $B \gets $ \textsc{pLeidenR}($G_l, C_l$)
    \State $G_{l+1} \gets$ \textsc{GraphContraction}($G_l$, $B$)
    \State $C_{l+1} \gets $ \textsc{Downsample}($G_{l+1}, C_l$)
    \State $l \gets l + 1$
\EndWhile
\While {$l - 1 \geq 0$} 
\State $l \gets l-1$
    \State $C_l \gets$ \textsc{ProjectCluster}($C_l, C_{l+1}$)
\EndWhile
\State $C_{out} \gets C_0$
\end{algorithmic}
\end{algorithm}

\subsection{Parallel Local Move Phase}
\label{section:alt_plm}
The only necessary trait of the local move heuristic for Leiden's proofs \cite{leiden_base} is that it return an improved clustering unless the input is node optimal.
Our standard parallel implementation cannot do this, so we provide an alternate one that does.
Essentially, we require a stricter version of the afterburner filter.
First, we gather all non-node-optimal vertices in some arbitrary order, and determine their ideal clusters.
Our stricter afterburner considers all prior vertices in the order, not just those immediately adjacent, which is necessary for lambdaCC objectives.
We determine the prefix of the order which maximizes the change in the objective function, and commit this prefix.
This maximizing prefix is non-empty if there are any non-node-optimal vertices.

This is possible in $O(n)$ total work with some complex usage of prefix sums, while the obvious $O(n^2)$ approach is much simpler.
We only invoke this alternate local move implementation when our standard local move implementation fails to improve the clustering (see line 6 of Algorithm \ref{alg:leiden}).
In our testing, the number of non-node-optimal vertices when this occurs was rarely more than 40.
Thus, the simpler $O(n^2)$ approach is likely faster than the more complex, higher overhead $O(n)$ approach.

\subsection{Parallel Refinement Heuristic}
\subsubsection{Original Leiden Refinement Heuristic}
We hereafter refer to Leiden's refinement heuristic as LeidenR; note that the Leiden authors denote it \textsc{RefinePartition}.
LeidenR is based upon cluster joining \cite{multilevel_cluster};
it performs a sequence of non-negative moves such that the size of its clusters (except for singleton clusters) strictly increase.
All optimal clusterings are reachable in this manner \cite{leiden_base}.
Beginning from a singleton clustering, LeidenR visits each vertex in a random order.
If it visits a vertex in a singleton cluster, it randomly chooses to move it to an adjacent cluster such that the objective does not decrease ($\gamma$-connectivity).
It can also randomly choose not to move the vertex.
Upon visiting a vertex not in a singleton cluster, it skips it.

LeidenR also utilizes a constraint clustering $B$, such that vertices can only move to clusters in their same constraint cluster.
Additionally, LeidenR ignores vertices not satisfying $w'(u,B[u]) \geq 0$, and also does not move vertices to clusters not satisfying $w'(C_x,B[C_x]) \geq 0$.
This is necessary for subpartition $\gamma$-density in stable iterations.

\subsubsection{Our Parallel LeidenR Algorithm (pLeidenR)}
Our parallel algorithm (algorithm \ref{alg:pleidenR}) first generates a forest of spanning trees; each tree induces a cluster.
It then gathers the vertices of each spanning tree and sorts them by an ordering.
Once sorted, we can ensure that the induced clusters satisfy LeidenR's aforementioned criteria.

To generate the spanning forest, we select a random edge $(u,v)$ from $u \in N(v) \cup v$ for each vertex $v$ such that $B[u]=B[v]$, $w'(u,v) \geq 0$, and $w'(u,B[u]) \geq 0$.
We then compute a random ordering function mapping each vertex to $[0, K)$ for some large integer $K \gg |V|$.
We subsample the edges going from a higher order vertex to a lower order vertex, to eliminate cycles.
For vertices in singleton trees, we modify their ordering by $f_{new}(v) = 2K - f_{old}(v)$.
This eliminates most of the singleton trees, except those vertices that selected self-loops.
Each tree has expected logarithmic depth due to the random ordering function.
We can extract them efficiently in parallel via pointer-chasing.
Sorting the vertices of each tree uses the same ordering function, so that each prefix of a sorted tree represents a connected tree.

\begin{lemma}
\label{lem:versa}
pLeidenR can generate a clustering only if LeidenR can generate it.
\end{lemma}
\begin{proof}
Assume that LeidenR visits the vertices $v \in V$ in increasing order of $f(v)$, and moves each $v$ to the cluster $C_x$ containing the lesser $f$-valued vertices in the same tree as $v$.
\textsc{CommitWellConnectedPrefixes} simulates LeidenR under this assumption, verifying if LeidenR could perform each move.
More specifically, it verifies $w'(v,C_x) \geq 0$ and $w'(C_x, B[C_x]) \geq 0$.
This requires two separate prefix sums, one over $w(C_x)$ and one over $w(C_x, B[C_x])$.
The longest prefix of valid moves of each tree is committed as a cluster; the remaining suffix is broken into singleton clusters (assume that LeidenR chooses not to move the suffix vertices).
Thus, by construction of the algorithm, the lemma is true.
\end{proof}

\begin{lemma}
\label{lem:vice}
pLeidenR can generate a clustering if LeidenR can generate it.
\end{lemma}
\begin{proof}
Consider the LeidenR algorithm.
For every vertex $v$, set $f(v) = i$, where $v$ is the $i$th vertex LeidenR will visit.
When LeidenR adds a vertex $v$ to a cluster $X$, choose an edge satisfying $w'(u,v) \geq 0$ with $u \in X$, and set $H[v] = (u,v)$.
Such an edge must exist, otherwise moving $v$ to $X$ has a negative effect on the objective.
Vertex $u$ is either already visited and therefore $f(u) < f(v)$, or will not be in a singleton-cluster when it is visited.
When LeidenR visits $u$ later, set $H[v] = (u,u)$ and $f(u) = 0$.
This ensures all chosen edges satisfy $f(u) < f(v)$.
By giving pLeidenR this $H$ and $f$, it will output the clustering generated by LeidenR.
Observe that pLeidenR may select this $H$ and $f$ randomly with a non-zero probability.
\end{proof}

\begin{theorem}
pLeiden provides all of the original Leiden algorithm's guarantees.
\end{theorem}
\begin{proof}
pLeidenR can generate a clustering if and only if LeidenR can generate it, thus they are equivalent up to their probabilities to generate a given clustering. As the original proofs of Leiden's guarantees \cite{leiden_base} are independent of the probability distribution over possible clusterings, replacing LeidenR with pLeidenR preserves their correctness. As we achieve the necessary traits of the local move heuristic as discussed in section \ref{section:alt_plm}, the theorem is true.
\end{proof}

\begin{algorithm}[htbp]
\caption{pLeidenR: The parallel refinement algorithm of pLeiden.}
\label{alg:pleidenR}
\begin{algorithmic}[1]
\Require $G = (V,E)$. Constraint Clustering $B$. $\mathit{seed}$.
\Ensure Vector $C_{out}$.
\State $H \gets $ $[null]\times|V|$
\State $U \gets \{u ~|~ u \in V \land w'(u, B[u]) \geq 0\}$
\For{$v \in U$ \textbf{in parallel}}
\State $A_v \gets \{(u,v) ~|~ u \in (N(v) \cup v) \cap  U \land B[u]=B[v] \land w'(u,v) \geq 0\}$
\State $H[v] \gets $ \textsc{ChooseRandom}($A_v)$
\EndFor
\State $f \gets $ \textsc{RandomizedOrderingFunction}($\mathit{seed}$)
\State $T \gets $ \textsc{ExtractTrees}($H, f$)
\State $T \gets $ \textsc{SerializeAndSort}($T,f$)
\State $C_{out} \gets $ \textsc{CommitWellConnectedPrefixes}($T$)
\end{algorithmic}
\end{algorithm}

\section{Empirical Results}
\subsection{Comparisons}
We compare pLouvain and pLeiden to v-Louvain (commit 91f2803) \cite{Sahu2025-ye}, GALA (commit e43bdea) \cite{wang_gpu}, cuGraph's Louvain \cite{cugraph_louvain} (v26.02), cuGraph's Leiden \cite{cugraph_leiden} (v26.02), GVE-Louvain (commit e228af3) \cite{gve_louvain}, GVE-Leiden (commit f741a40) \cite{leiden_sahu}, and Networkit Louvain (v11.1.post1) and Leiden (v11.2) \cite{nguyen_leiden}.
We are unable to compare to Parallel Correlation Clustering \cite{shi2021scalable}, as their code could not build successfully.
GVE-Louvain outperforms all CPU and GPU parallel Louvain implementations to which its authors compared.
We consider this test set to be representative of the current state-of-the-art in shared-memory parallel Louvain and Leiden implementations.

We add an uncoarsening phase to pLeiden to produce pLeiden+.
Due to the uncoarsening phase, pLeiden+ has no per-iteration guarantees; $\gamma$-separation and $\gamma$-connectivity are instead guaranteed only in stable iterations.

\subsection{Test Systems}
We assess our implementations, as well as v-Louvain, GALA, and cuGraph, on an Nvidia B200 GPU with 180GB of VRAM.
We assess GVE-Louvain/Leiden, and NetworKit Louvain/Leiden on an AMD Ryzen 9950x3D CPU, and run each with 32 threads.
This system has 96GB of 6GT/s DDR5 RAM in a dual-channel configuration.
We compile GPU programs with Cuda Toolkit version 12.8 (driver version 580.95.05) and CPU programs with g++ 13.3.0.

\subsection{Test Data and Configuration}
We use a set of graphs~\cite{Gilbert2024-se} previously used to evaluate graph partitioning methods on GPUs~\cite{jet}.
This test set overlaps with the test sets of competitor works \cite{Sahu2025-ye, wang_gpu, leiden_sahu}, but with some additional preprocessing.
Each graph is symmetrized, with all edge weights set to 1, and only the largest connected components retained.
We run each program on each graph 21 times (except for cugraph which we run 5 times), and collect the median clustering time and modularity result. The synchronous implementations also lead to low runtime variance. 
We measure quality in terms of modularity as it is the most common objective function used to evaluate Louvain and Leiden implementations, and the other programs only have the ability to optimize for modularity. For our programs, we compute the standard deviation in modularity, and find it to be 0.0002 on average. 

\begin{table}[htbp]
    \centering
    \caption{A Collection of undirected graphs used for performance evaluation. We preprocess the graphs to extract the largest connected component. The number of vertices ($n$), edges ($m$), the ratio of max vertex degree ($\Delta$) to average degree after preprocessing (rounded to nearest integer), and the modularity $Q$ with a baseline sequential Louvain+ implementation are given. ($\ddagger$: Use a ``synchronous'' version.)}
    \begin{tabular}{@{}lS[table-format=3.1]S[table-format=3.1]S[table-format=6,round-mode=places,round-precision=0]S[table-format=0.4,round-mode=places,round-precision=4]@{}}
    \toprule
        Graph & {$n$ ($\times 10^6$)} & {$m$ ($\times 10^6$)} & {$\Delta/(2m/n)$} & $Q$\\
        \midrule
grid & 8.0 & 16.0 & 1 & 0.990263\\
cube & 8.0 & 23.9 & 1.01 & 0.961124\\
delaunay23 & 8.4 & 25.2 & 4.67 & 0.989438\\
bubbles00 & 18.3 & 27.5 & 1 & 0.993844\\
bubbles10 & 19.5 & 29.2 & 1 & 0.993773\\
rgg22 & 4.2 & 30.4 & 2.49 & 0.986604\\
bubbles20 & 21.2 & 31.8 & 1.00 & 0.994043\\
delaunay24 & 16.8 & 50.3 & 4.33 & 0.991533\\
rgg23 & 8.4 & 63.5 & 2.64 & 0.988320\\
rgg24 & 16.8 & 132.6 & 2.53 & 0.989927\\[4pt]
ppa & 0.6 & 21.2 & 43.97 & 0.742384\\
cage15 & 5.2 & 47.0 & 2.52 & 0.894352\\[4pt]
kmerV2a & 53.5 & 57.1 & 18.28 & 0.991366$^\ddagger$\\
kmerU1a & 64.7 & 66.4 & 17.05 & 0.989436$^\ddagger$\\
kmerP1a & 138.9 & 148.5 & 18.71 & 0.976105$^\ddagger$\\
kmerA2a & 170.4 & 179.9 & 18.94 & 0.975496$^\ddagger$\\
kmerV1r & 214.0 & 232.7 & 3.68 & 0.953673$^\ddagger$\\[4pt]
feRotor & 0.1 & 0.7 & 9.4 & 0.908991\\
afShell & 1.5 & 25.6 & 1 & 0.887928\\
Hook1498 & 1.5 & 29.7 & 2.32 & 0.893854\\
Geo1438 & 1.4 & 30.9 & 1.3 & 0.873779\\
Serena & 1.4 & 31.6 & 5.46 & 0.877983\\
audikw & 0.9 & 38.4 & 1.59 & 0.912493\\
channel050 & 4.8 & 42.7 & 1.01 & 0.851467\\
LongCoup & 1.5 & 42.8 & 12.96 & 0.863976\\
dielFilterV3 & 1.1 & 44.1 & 3.36 & 0.926019\\
MLGeer & 1.5 & 54.7 & 1 & 0.818947\\
Flan1565 & 1.6 & 57.9 & 1.08 & 0.841761\\
Bump2911 & 2.9 & 62.4 & 4.43 & 0.857308\\
CubeCoup & 2.2 & 62.5 & 1.16 & 0.850907\\
HV15R & 2.0 & 162.4 & 3.06 & 0.835900\\
Queen4147 & 4.1 & 162.7 & 1.02 & 0.727577\\[4pt]
nlpkkt120 & 3.5 & 46.7 & 1.03 & 0.924715\\
nlpkkt160 & 8.3 & 110.6 & 1.02 & 0.936780\\
nlpkkt200 & 16.2 & 216.0 & 1.02 & 0.944245\\[4pt]
roadUSA & 23.9 & 28.9 & 3.73 & 0.998156\\
europeOsm & 50.9 & 54.1 & 6.12 & 0.999004\\[4pt]
circuit5M & 5.6 & 27.0 & 132863.56 & 0.817375$^\ddagger$\\[4pt]
vasStokes2M & 2.1 & 48.4 & 29.01 & 0.891383\\
vasStokes4M & 4.3 & 97.7 & 25.32 & 0.911944\\
stokes & 11.3 & 258.0 & 37.86 & 0.929308$^\ddagger$\\[4pt]
dblp10 & 0.2 & 0.7 & 37.61 & 0.864228\\
amazon08 & 0.7 & 3.5 & 112.38 & 0.893075\\
socPokec & 1.6 & 22.3 & 543.76 & 0.690303\\
citation & 2.9 & 30.3 & 480.37 & 0.833754\\
comLiveJournal & 4.0 & 34.7 & 853.92 & 0.723462\\
socLiveJournal & 4.8 & 42.8 & 1149.38 & 0.731914\\
ljournal08 & 5.4 & 49.5 & 1052.40 & 0.766239\\
hollywood09 & 1.1 & 56.3 & 108.87 & 0.751641\\
products & 2.4 & 61.8 & 337.41 & 0.878197\\
hollywood11 & 1.9 & 114.3 & 109.94 & 0.753702\\
Orkut & 3.1 & 117.2 & 436.71 & 0.663088\\
enwiki21 & 6.3 & 136.5 & 5324.22 & 0.663872\\[4pt]
wbEdu & 8.9 & 44.2 & 2585.75 & 0.990678\\
ic04 & 7.3 & 149.1 & 6296.91 & 0.963015\\
uk02 & 18.5 & 261.6 & 6879.39 & 0.990040\\
arabic05 & 22.6 & 552.2 & 11797 & 0.989574\\
\bottomrule
   \end{tabular}
\label{tab:graphs}
\end{table}

For pLouvain, pLeiden, and pLeiden+, we set the local move heuristic pass limit to 6, with the temperature parameter $\phi = 0.75$ for the first four passes and $\phi = 0.25$ for the last two.
We leave all settings as default for the other programs, except for NetworKit Leiden where we set the Leiden run count to 1 for consistency with other Leiden implementations.
We use the modularity numbers given by the output of each program where applicable and accurate (Our Louvain/Leiden, v-Louvain, GVE-Louvain/Leiden, cugraph Louvain, GALA), or compute the modularity values where none are given (NetworKit Louvain/Leiden) or are incorrect (cuGraph Leiden).

In Fig.~\ref{fig:runtime_all}, we present the geometric mean across all graphs of single-iteration runtimes, normalized by the runtime of pLouvain.
In Fig.~\ref{fig:modularity_all}, we present the average modularity delta from pLouvain across all graphs.
In Fig.~\ref{fig:runtime_all} and Fig.~\ref{fig:modularity_all}, we exclude instances where the competitor programs ran out of memory (GVE-Louvain and GVE-Leiden on one graph, Networkit Louvain on two graphs) or crashed (GALA on three graphs) from the calculations for those programs.
For Fig.~\ref{fig:modularity_all}, we additionally exclude substantial outliers, which differ from pLouvain by more than 0.2; this affects 2 graphs for GALA, and 1 graph each for cugraph Louvain and Leiden.
In Fig.~\ref{fig:runtime_modularity_iterated}, we plot the runtime and modularity improvement of performing 10 additional (11 total) iterations of pLouvain, pLeiden, and pLeiden+ on the test set.
The programs do not necessarily reach stability in this number of iterations.
We include modularity results of clusterings obtained from VieClus \cite{BiedermannHSS18} by running on 24 processes with a 30-minute time limit, one trial per graph.
We run this experiment on a system with 180 vCPUs (from an AMD EPYC 9655) with 740 GB of RAM; the process count is chosen to avoid running out of memory for most graphs.
A recent work \cite{ansotegui2025uncovering} showed that VieClus found clusterings of optimal modularity for every graph that they tested, and in less than 25 seconds per graph.
Our test graphs are much larger, by six orders of magnitude in some cases, so it is unlikely that VieClus can find the optimal clusterings within a reasonable timeframe.
Regardless, VieClus serves as a compelling benchmark to compare with our iterated methods.

In Table~\ref{tab:ablation}, we conduct an ablation study of several key aspects of pLouvain, to determine their effectiveness.
We showcase the following experiments:
\begin{enumerate}
    \item No local move heuristic passes in the uncoarsening phase (i.e. standard Louvain)
    \item Standard Louvain, but with double the local move heuristic passes in the coarsening phase
    \item No symmetry breaking (afterburner filter disabled and $\phi = 0$)
    \item Simulated annealing disabled (i.e. $\phi = 0$)
    \item Symmetry breaking with the MLH instead of the afterburner filter. MLH is active in odd passes, and inactive in even passes. $\phi = 0$
    \item Fused enumeration and argmax kernels in passes when many vertices ($>3$\%) are moved
    \item Fused enumeration and argmax kernels in all passes
    \item Gathering vertices by cluster membership in contraction
\end{enumerate}

\begin{table}
    \centering
    \caption{Ablation Study of pLouvain. Slowdown is calculated as the geometric mean of running time ratio compared to the default pLouvain run (lower values means faster). The modularity delta gives the average difference from the baseline (higher is better). }
    \begin{tabular}{@{}lS[table-format=1.3{$\times$}]S[table-format=+1.4]@{}}
     \toprule
       Experiment & {Slowdown} & {Mod. delta} \\
       \midrule
       Default pLouvain & 1               & 0\\
       Louvain (i.e., no uncoarsening) & 0.775{$\times$} & -0.0126 \\
       Louvain, 2x LMH Passes & 1.16{$\times$} & -0.0082 \\
       pLouvain, no symmetry breaking + $\phi = 0$ & 1.407{$\times$} & -0.4133 \\
       pLouvain, no simulated annealing (i.e. $\phi = 0$) & 0.982{$\times$} & -0.0014 \\
       pLouvain, MLH + $\phi = 0$ & 0.94{$\times$} & -0.0017 \\
       pLouvain, fused kernel ($>3\%$ vertices moved) & 1.001{$\times$} & 0.0001 \\
       pLouvain, fused kernel (all passes) & 1.058{$\times$} & 0.0001 \\
       pLouvain, gathered contraction & 1.371{$\times$} & 0 \\
       \bottomrule
    \end{tabular}
    \label{tab:ablation}
\end{table}

\begin{figure}
    \centering
    \includegraphics[width=\linewidth]{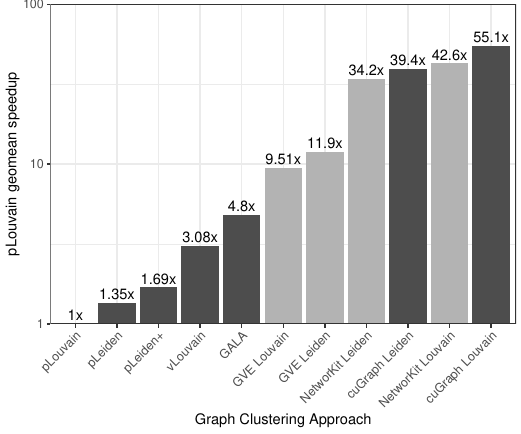}
    \caption{Running time comparison of clustering algorithms.}
    \label{fig:runtime_all}
\end{figure}

\begin{figure}
    \centering
    \includegraphics[width=\linewidth]{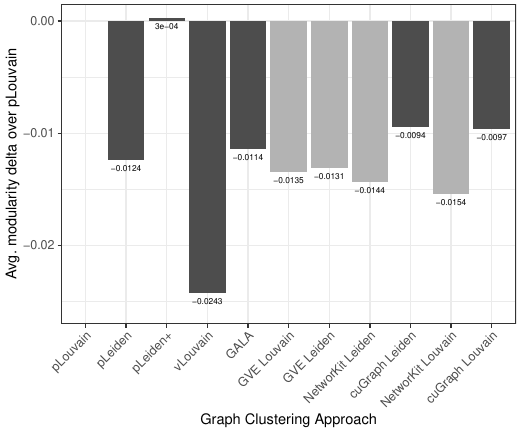}
    \caption{Modularity comparison of clustering algorithms.}
    \label{fig:modularity_all}
\end{figure}

\begin{figure}
    \centering
    \includegraphics[width=0.49\linewidth]{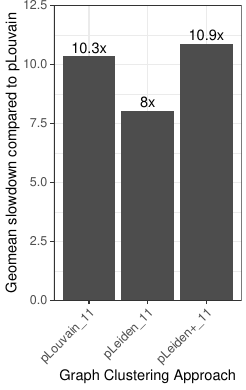}%
\includegraphics[width=0.49\linewidth]{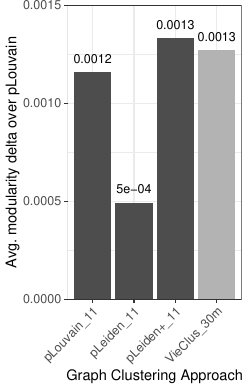}
    \caption{Running time (left) and modularity (right) comparisons with 10 additional (11 total) iterations.}
    \label{fig:runtime_modularity_iterated}
\end{figure}


\section{Analysis}
\subsection{Comparison to Parallel State of the Art}
\subsubsection{Runtime}
The longest runtime of pLouvain on any graph is just 0.5974 seconds.
It is the fastest method on 55 of 57 graphs; the two smallest test graphs account for both cases when pLouvain is not the fastest.
It achieves a processing rate of up to 6.1 billion nonzeroes per second ($\frac{2m}{t}$).
pLeiden is 35.3\% slower by geometric mean than pLouvain; it is faster than every competitor (Louvain or Leiden) on 52 of 57 graphs.
pLeiden+ is 25.1\% slower by geometric mean than pLeiden, and is faster than every competitor on 43 of 57 graphs.

pLeiden and pLeiden+ are slower than pLouvain in part due to the additional cost to compute pLeidenR (13.8\% of pLeiden's runtime on average), and in part due to the slower coarsening rate than pLouvain.
pLeiden and pLeiden+ process 29.1\% more edges by geometric mean as a result of their slower coarsening rate than pLouvain.
This trend is common to any implementation of Leiden; the original Leiden algorithm only manages to be faster than the original implementation of Louvain due to vertex pruning techniques and early termination within the local move heuristic.
These techniques are commonplace in the implementations of Louvain in our comparison set, and within our own implementations.

\subsubsection{Modularity}
pLouvain and pLeiden+ both achieve higher modularity than any competitor on 56 of 57 graphs, and the best overall method on every graph is either pLouvain or pLeiden+.
pLeiden achieves higher modularity than any competitor on only 14 of 57 graphs, and overall performs in the middle of the pack.
pLeiden performs worse than pLouvain and pLeiden+ due to the lack of an uncoarsening phase and a local move heuristic pass count that is tuned for the presence of uncoarsening.

v-Louvain's modularity results are worse than GVE-Louvain by a wide margin;
both are asynchronous algorithms with largely similar design choices.
This demonstrates that asynchronous approaches scale poorly to the massive concurrency of the B200 GPU.
GALA performs better in terms of modularity vs v-Louvain.
However, it is substantially worse than all competitors on the grid graph and circuit5M, two outliers which we exclude from GALA's Fig.~\ref{fig:modularity_all} data.
It is synchronous, but degree-based batching is its only form of symmetry-breaking.
This is ineffective when one batch contains almost all vertices;
99.9\% of vertices in the grid graph and 92.9\% of vertices in circuit5M have degree 4.

\subsection{Ablation Study}
\label{section:ablation}
pLouvain is both extremely efficient and very high quality.
Table~\ref{tab:ablation} provides insight into the contribution of certain design choices towards this result.
Without any symmetry breaking, pLouvain converges 41\% slower to drastically worse solutions.
Disabling simulated annealing leads to worse solutions by one sixth of the difference between pLouvain and the closest competitor (cugraph Leiden), with a 1.8\% decrease in runtime.
Using MLH in place of the afterburner filter (thus losing compatibility with the simulated annealing techniques of \cite{dkaminpar_jet}) similarly leads to worse solutions by one fifth of the difference between pLouvain and its nearest competitor,
but also brings slightly faster runtimes by 6.4\% as a tradeoff.
We observe that 80\% of the quality benefit of the afterburner filter versus MLH is because it enables simulated annealing, whereas simulated annealing contributes a minority of the slowdown versus MLH.

Our results show little difference between kernel fission and fusion for passes when the number of moves is large.
In contrast, we observe a positive effect ($5.8\%$ speedup) with kernel fission for small move sets.
Kernel fission permits a more efficient enumeration update kernel for small move sets,
which are very common in the uncoarsening phase.

Graph contraction takes just 6.5\% of pLouvain's runtime on average and exceeds 10\% of the runtime on just two graphs.
At most, it accounts for 12.6\% of the total pLouvain runtime.
Versus the vertex-gathering-based contraction ablation experiment, pLouvain is 37.1\% faster, and shows a 3.85x speedup for contraction specifically.
pLouvain's graph contraction is 15.5x faster than that of v-Louvain by geometric mean, and 15.1x faster than that of GALA.

Uncoarsening is the most important reason for pLouvain's high quality.
By performing the local move heuristic within the uncoarsening phase, the runtime increases by just 29.0\%, but the modularity increases by 0.0127, which is more than the difference between pLouvain and its nearest competitor (as well as the difference to pLeiden).
It effectively performs twice the local move heuristic passes that standard Louvain would with the same pass limit, therefore we compare to standard Louvain with double the pass limit.
That configuration is 16.0\% slower than pLouvain, with 0.0082 lower modularity on average (0.0012 higher than cugraph Leiden).
More local move heuristic passes can be performed more cheaply and to greater effect in the uncoarsening phase.
This is especially important for synchronous algorithms that require more passes than asynchronous algorithms.

\subsection{Iterating our Methods}
pLeiden is slower and lower quality than pLouvain, and pLeiden+ is slower but of roughly similar quality to pLouvain.
However, Leiden's greatest strengths are the additional guarantees it obtains over successive iterations.
\subsubsection{Modularity}
Fig.~\ref{fig:runtime_modularity_iterated} shows that pLeiden, across 11 total iterations, produces lower modularity results than pLouvain with the same iteration count.
This is a consequence of no uncoarsening in pLeiden, which is consistent with the single iteration experiments.
In contrast, iterated pLeiden+ produces superior results to iterated pLouvain, and on average very slightly outperforms VieClus given 30 minutes with 24 processes.

\subsubsection{Runtime}
While pLeiden is slower than pLouvain for the first iteration, Fig.~\ref{fig:runtime_modularity_iterated} shows that 11 total iterations of pLeiden are 28.8\% faster by geomean than the same iteration count for pLouvain.
As pLouvain begins from a singleton clustering on all iterations, successive pLouvain iterations see only slight speedups versus the first iteration.
pLeiden begins from the given input clustering, therefore the local move heuristic has much less work to do in successive iterations.
Iterated pLeiden+ is 36.3\% slower than iterated pLeiden, a larger difference than between their first iteration runtimes.
This is likely because there is a smaller difference between local move heuristic passes in the coarsening versus the uncoarsening for successive pLeiden+ iterations.

\subsubsection{Versus VieClus}
Iterated pLeiden+ exceeds or matches VieClus' quality on 26 of 57 graphs.
Iterated pLouvain exceeds or matches VieClus' quality on 23 graphs.
We expect that VieClus would eventually beat iterated pLeiden+ on these graphs given enough time, but the time requirement for that to occur may be significant.
Given the runtime constraints of our experiments, VieClus is limited by the speed with which it computes initial population and offspring clusterings.

\section{Conclusion}
We develop pLouvain and pLeiden, novel GPU-parallel implementations of Louvain+ and Leiden, respectively. Both benefit from our high-performance synchronous implementation of the local move heuristic. We find that uncoarsening is vital given the synchronous local move.
pLeiden provably ensures each of the six original Leiden guarantees.

We conduct an extensive empirical evaluation with 57 test graph instances from 10 families. Our performance results are obtained primarily on an NVIDIA B200 GPU and a 16-core AMD Ryzen 9950x3D CPU.

Our graph contraction scheme achieves substantial speedups over comparable GPU-parallel methods. Optimizations such as the afterburner filter for symmetry breaking, kernel fission, and gather-free contraction, contribute to improvements in runtime or quality. Our parallel implementation of the Leiden refinement scheme accounts for only a small fraction of pLeiden's running time.

Our novel iteration scheme for Louvain+ guarantees basic internal connectivity of clusters upon reaching stability. We additionally construct a GPU-parallel uncoarsening extension to Leiden, pLeiden+. 
It also ensures the Leiden guarantees, except the per-iteration guarantees are delayed until stable iterations.

With just 10 iterations, pLouvain and pLeiden+ are both competitive with the memetic clustering algorithm VieClus. There is thus potential for a memetic clustering algorithm on the GPU based around pLouvain or pLeiden+.
Another potential use-case for our programs is layered-label propagation for graph compression \cite{BRSLLP}.

\section*{Acknowledgements}
This research is supported in part by NSF grants 1955971 and 2437873. We thank the reviewers for their helpful comments.

\bibliographystyle{ieeetr}
\bibliography{refs}






\twocolumn[%
{\begin{center}
\Huge
Appendix: Artifact Description        
\end{center}}
]



\section{Overview of Contributions and Artifacts}

\subsection{Paper's Main Contributions}

\artexpl{
Provide a list of all main contributions of the paper.
}

\begin{description}
    \item[$C_1$] On an NVIDIA B200 GPU and a 57-instance graph dataset, pLouvain achieves a 3.1x geometric mean speedup over v-Louvain.
    \item[$C_2$] pLeiden offers quality guarantees equivalent to Leiden, and is 8.8x faster (geometric mean) than a prior parallelization.
    \item[$C_3$] pLouvain and pLeiden+ consistently outperform other approaches in terms of quality (modularity).
    \item[$C_4$] We present an efficient symmetry-breaking technique based on the Jet graph partitioner's \textit{afterburner} filter.
    \item[$C_5$] Our graph contraction implementation is $\geq$15x faster by geometric mean than state-of-the-art competitors v-Louvain and GALA.
    \item[$C_6$] We give an improved iteration scheme to mitigate Louvain's weak internal cluster connectivity problem.
\end{description}

\subsection{Computational Artifacts}

\artexpl{
List the computational artifacts related to this paper along with their respective DOIs. Note that all computational artifacts may be archived under a single DOI.
}

\begin{description}
\item[$A_1$] https://doi.org/10.5281/zenodo.18717935
\item[$A_2$] https://doi.org/10.5281/zenodo.18719087
\end{description}

\artexpl{
Provide a table with the relevant computational artifacts, 
highlight their relation to the contributions (from above) and 
point to the elements in the paper that are reproducible by each artifact, e.g., 
which figures or tables were generated with the artifact.
}

\begin{center}
\begin{tabular}{rll}
\toprule
Artifact ID  &  Contributions &  Related \\
             &  Supported     &  Paper Elements \\
\midrule
$A_1$   &  $C_{1-6}$ & Table 2 \\
        &        & Figures 1-3\\
\midrule
$A_2$   &  $C_{1-6}$ & Table 2 \\
        &        & Figures 1-3\\
\bottomrule
\end{tabular}
\end{center}

\section{Artifact Identification}

\artexpl{
Provide the following six subsections for each computational artifact $A_i$.
}

\newartifact

\artrel
This artifact contains the code for pLouvain, pLeiden, pLeiden+, and the ablation experiments.
The ablation experiments are contained in appropriately named branches in the git repository linked from the DOI.
This artifact is necessary for experiments showcased in Table 2 and Figures 1-3.

\arttime
Compilation should take 5-10 minutes, including dependencies.

\artin

\artinpart{Hardware}
The programs should run on any Nvidia GPU since the Turing architecture.
It has been tested on a B200, RTX 4090, and an RTX 5090.

\artinpart{Installation and Deployment}
The code depends on the Kokkos (https://github.com/kokkos/kokkos) framework, version $\geq$4.7.0.
It also depends on KokkosKernels (https://github.com/kokkos/kokkos-kernels), version $\geq$4.7.0.
To compile, Cuda Toolkit version $\geq$12.0 and CMake version $\geq$3.28 are required.

\newartifact

\artrel

\artexpl{
    Briefly explain the relationship between the artifact and contributions.
}
This artifact contains scripts to generate all of the running time and clustering modularity results underlying Table 2, and Figures 1-3, as well as some miscellaneous quantities in section VII.
Additionally, it contains the data produced by our experiments, organized into 3 levels of granularity.
At the finest level, we have individual running times and modularity results for each combination of graph and program.
At the next level, we aggregate these results by Table/Figure.
At the coarsest level, we have the headline numbers used to create each Table/Figure.

\artexp

\artexpl{
Provide a higher level description of what outcome to expect from the corresponding experiments. Provide an explanation of how the results substantiate the main contributions.
}

\artsampl{
Algorithm A should be faster than Algorithms C and B in all GPU scenarios.    
}

The geometric mean runtimes should be, from shortest to longest: pLouvain, pLeiden, pLeiden+, v-Louvain, GALA, GVE Louvain, GVE Leiden, Networkit Leiden, cuGraph Leiden, Networkit Louvain, and cuGraph Louvain.
The average modularity results should be (outliers excluded), from highest to lowest: pLeiden+, pLouvain, cuGraph Leiden, cuGraph Louvain, GALA, pLeiden, GVE Leiden, GVE Louvain, Networkit Leiden, Networkit Louvain, and v-Louvain.

For the multiple iteration experiments, pLeiden\_11 should be fastest but lowest quality,  pLeiden+\_11 should be slowest but highest quality, and pLouvain\_11 should be in the middle for both runtime and quality.

VieClus quality results should be slightly lower than pLeiden+\_11 quality.

The results obtained by running our scripts should be similar to our results in idpds26\_results.tar.gz.

\arttime

\artexpl{
Estimate the time required to reproduce the artifact, providing separate estimates for the individual steps: Artifact Setup, Artifact Execution, and Artifact Analysis.
}

\artsampl{
The expected computational time of this artifact on GPU X is 20~min.    
}
Compilation of ours and competitor code should take up to 20 minutes.
Computational time for GPU experiments including competitors should take 6-8 hours.
Computational time for CPU competitors should take 6-8 hours.
Computational time for VieClus should take about 2-3 days.

\artin

\artinpart{Hardware}

\artexpl{
Specify the hardware requirements and dependencies (e.g., a specific interconnect or GPU type is required).
}
GPU experiments are run on the Nvidia B200, rented from cloud provider Verda (https://verda.com/).
Runtime results vary to a small degree ($<$5\%) depending on which GPU you get assigned.
We believe this is because some instances have slower CPUs than others, leading to longer kernel launch latencies.
All of our GPU results, except for the "gathered contraction" ablation experiment, were generated on the exact same instance.
For that specific experiment, we reran the baseline pLouvain experiment on the same instance for normalization purposes.

CPU experiments are run on a consumer grade system having an AMD Ryzen 9950x3d 16-core CPU and 96GB of DDR5 dual-channel memory.

The VieClus experiment is run on a 180-vcore cloud instance from Verda.

\artinpart{Software}

\artexpl{
Introduce all required software packages, including the computational artifact. For each software package, specify the version and provide the URL.
}
Our code is detailed in artifact $A_1$.

Competitor Software:
\begin{enumerate}
    \item v-Louvain (https://github.com/puzzlef/louvain-communities-cuda)
    \item GALA
    \item GVE Louvain (https://github.com/puzzlef/louvain-communities-openmp)
    \item GVE Leiden (https://github.com/puzzlef/leiden-communities-openmp)
    \item Networkit Louvain v11.1.post1 and Leiden v11.2 (https://networkit.github.io/)
    \item cuGraph Louvain and Leiden (https://github.com/rapidsai/cugraph) v26.02
    \item VieClus v1.1 (https://github.com/VieClus/VieClus)
\end{enumerate}

\artinpart{Datasets / Inputs}

\artexpl{
Describe the datasets required by the artifact. Indicate whether the datasets can be generated, including instructions, or if they are available for download, providing the corresponding URL.
}
The test graph dataset is available at https://scholarsphere.psu.edu/resources/fd9ba209-a0cd-4f33-994b-c22ae3bcb243/downloads/35163?download=true.
It contains a superset of the graphs listed in Table 1, represented in the Metis file format.
The "data\_setup.sh" script handles the process of downloading, decompressing, and de-archiving the dataset.
It also creates matrix marketplace format copies of the graph files.

\artinpart{Installation and Deployment}

\artexpl{
Detail the requirements for compiling, deploying, and executing the experiments, including necessary compilers and their versions.
}
We compiled ours and competing software with Cuda Toolkit version 12.8 and g++ 13.3.0.
For VieClus, we used OpenMPI version 4.1.6.
Run the "setup\_gpu.sh" script followed by the "setup\_our\_code.sh" script to install our code's dependencies and build our code.
Run the "setup\_competitors\_gpu.sh" script to build v-Louvain and GALA.
Run the "setup\_cpu\_competitors.sh" to build GVE-Louvain and GVE-Leiden.
The latter two scripts clone our forks of these programs, which include an improved data collection process and some bugfixes for GALA.

For the cuGraph and Networkit programs, you will first need to install conda/miniconda, then install rapidsai::cugraph and conda-forge::networkit.

For VieClus, you can run "setup\_and\_run\_vieclus.sh" which handles both the program building and experiment execution.

\artcomp

\artexpl{
Provide an abstract description of the experiment workflow of the artifact. It is important to identify the main tasks (processes) and how they depend on each other. 

A workflow may consist of three tasks: $T_1, T_2$, and $T_3$. The task $T_1$ may generate a specific dataset. This dataset is then used as input by a computational task $T_2$, and the output of $T_2$ is processed by another task $T_3$, which produces the final results (e.g., plots, tables, etc.). State the individual tasks $T_i$ and provide their dependencies, e.g., $T_1 \rightarrow T_2 \rightarrow T_3$.

Provide details on the experimental parameters. How and why were parameters set to a specific value (if relevant for the reproduction of an artifact), e.g., size of dataset, number of data points, input sizes, etc. Additionally, include details on statistical parameters, like the number of repetitions.
}

First, you must clone the appropriate git repositories, compile code, and download/uncompress the graph dataset.
You can do this with the aforementioned scripts.
Second, you will run our experiment scripts, which run our code and the competitors on each graph, and output the runtime/modularity datapoints.
These scripts include: "run\_our\_experiments.sh", "run\_gpu\_competitor\_experiments.sh", "run\_cpu\_competitor\_experiments.sh", "conda\_gpu\_competitors.sh", "conda\_cpu\_competitors.sh", and "setup\_and\_run\_vieclus.sh".

\artout
In our analysis, we aggregate the raw datapoints into spreadsheets by Table/Figure.
This is a manual process, which involves copying columns from the results files and pasting them into a spreadsheet.
We collect headline numbers from these spreadsheets, including the geometric mean runtimes and mean modularity deltas from pLouvain.
The headline numbers are the basis of Table 2 and Figures 1-3.

\end{document}